\documentclass[fleqn,10pt]{wlscirep}
\usepackage{graphicx}
\usepackage{amsmath}
\usepackage{amssymb}
\usepackage{amsfonts}
\usepackage{mathrsfs}
\usepackage{amsthm}
\usepackage{bm}
\usepackage{url}
\newtheorem{theorem}{Theorem}

\title{Frobenius-norm-based measures of quantum coherence and asymmetry}

\author[1,2,$\ast$,$\dagger$]{Yao Yao}
\author[3,$\ast$]{G. H. Dong}
\author[4]{Xing Xiao}
\author[3,$\dagger$]{C. P. Sun}
\affil[1]{Microsystems and Terahertz Research Center, China Academy of Engineering Physics, Chengdu Sichuan 610200, China}
\affil[2]{Institute of Electronic Engineering, China Academy of Engineering Physics, Mianyang Sichuan 621999, China}
\affil[3]{Beijing Computational Science Research Center, Beijing, 100193, China}
\affil[4]{College of Physics and Electronic Information, Gannan Normal University, Ganzhou Jiangxi 341000, China}
\affil[$\dagger$]{Corresponding author: yaoyao@mtrc.ac.cn,cpsun@csrc.ac.cn}
\affil[$\ast$]{These authors contributed equally to this work}

\begin{abstract}
We formulate the Frobenius-norm-based measures for quantum coherence and asymmetry respectively.
In contrast to the resource theory of coherence and asymmetry, we construct a natural measure of quantum
coherence inspired from optical coherence theory while the group theoretical approach is employed
to quantify the asymmetry of quantum states. Besides their simple structures and explicit physical
meanings, we observe that these quantities are intimately related to the purity (or linear
entropy) of the corresponding quantum states. Remarkably, we demonstrate that the proposed coherence quantifier
is not only a measure of mixedness, but also an intrinsic (\textit{basis-independent}) quantification
of quantum coherence contained in quantum states, which can also be viewed as a normalized version of
Brukner-Zeilinger invariant information. In our context, the asymmetry of $N$-qubit
quantum systems is considered under local \textit{independent} and \textit{collective} $\mathbb{SU}(2)$ transformations.
Intriguingly, it is illustrated that the collective effect has a significant impact on the asymmetry measure,
and quantum correlation between subsystems plays a non-negligible role in this circumstance.
\end{abstract}

\begin{document}

\flushbottom
\maketitle

\thispagestyle{empty}

\section*{Introduction}
In the 20th century, quantum mechanics, as the core of quantum physics, is undoubtedly one of the most profound scientific theories
during the development process of modern science and philosophy. Notably, quantum coherence is one of the most remarkable
and characteristic traits of quantum mechanics and also viewed as the critical resource for the emerging field of quantum technologies,
such as quantum cryptography~\cite{Gisin2002}, quantum metrology~\cite{Giovannetti2011} and quantum computation~\cite{Nielsen2000}.
This intrinsic principle of quantum mechanics enforces its entire departure from classical lines of thought,
similar to the notion of quantum entanglement~\cite{Schrodinger1935}. In different occasions, the generalized concept of quantum coherence
manifests itself as quantum superposition, quantum asymmetry or non-commutability. These specific notions have gotten extensive applications
in quantum physics both theoretically and experimentally. Therefore, the characterization, quantification and application
of quantum coherence are one of the central topics in quantum information science.

Although the investigation of coherence theory has a long history in classical and quantum optics~\cite{Glauber1963,Sudarshan1963,Mandel1995,Born1999,Wolf2007},
it does not provide an integrated and unified framework for quantitative study. In some contexts, even conflicting
definitions have appeared with respect to the same concept~\cite{Karczewski1963,Wolf2003,Tervo2003,Setala2004a,Wolf2004,Setala2004b},
especially when the vectorial character of electromagnetic waves is involved. From the perspective of quantum information theory,
a rigorous framework for the characterization and quantification of coherence has been formulated very recently, based on
the quantum resource theory~\cite{Horodecki2013,Brandao2015}. On one hand, Baumgratz \textit{et al.} defined the incoherent states and
incoherent operations and further discussed the quantum coherence monotones in the constraint of a series of
axiomatic-like prerequisites~\cite{Baumgratz2014}, which is extremely similar to the approach adopted in
quantum entanglement theory~\cite{Vedral1997}. Within this framework, novel measures of quantum coherence have been proposed and
the connections between coherence and other manifestations of quantum correlations have been carefully
scrutinized~\cite{Girolami2014,Shao2015,Streltsov2015,Yao2015,Yuan2015,Xi2015,Cheng2015,Rana2016,Rastegin2016,Napoli2016,Chitambar2016a,Winter2016}.
On the other hand, Marvian and his collaborators proposed the resource theory of asymmetry which included
the former notion of quantum coherence as a special case~\cite{Gour2008,Gour2009,Marvian2013,Marvian2014a,Marvian2014b,Marvian2015}.
Moreover, the internal relations between these two approaches and their physical justifications have been further clarified lately~\cite{Marvian2016,Chitambar2016b}.

However, the story has never come to the end. The investigation of quantum coherence is now facing several crucial problems that need to be solved urgently.
First of all, we should further identify the application scopes and physical meanings of existing measures of quantum coherence. In particular, the quantifiers
presented in ref. \cite{Baumgratz2014} are established in the finite-dimensional setting and thus can not be directly applied to the infinite-dimensional systems.
Moreover, explicit physical (or operational) meanings of existing measures of quantum coherence are to be pursued by the community, though partial results
have been obtained \cite{Winter2016,Singh2015a,Bagan2015}. Secondly, as indicated by Marvian \textit{et al.},
distinct approaches of quantifying quantum coherence will lead to different conclusions in a diverse set of tasks. For instance,
the \textit{speakable} notions of coherence may not be suitable to characterize the consumed resource in the task of phase estimation \cite{Marvian2015,Marvian2016}.
Moreover, Chitambar \textit{et al.} also pointed out that all currently proposed \textit{basis-dependent} theories fail to satisfy a criterion of physical
consistency \cite{Chitambar2016b}. In addition, the interactions between quantum coherence and other manifestations of quantum correlations still
need to be uncovered \cite{Streltsov2015,Yao2015}. Finally, the quantitative relations between quantum coherence and other important
quantities in quantum information are yet to be established \cite{Singh2015b,Bera2015}.

In contrast to the resource-theory framework, we notice that other types of coherence or asymmetry measures are also of great significance
in optical coherence theory and condensed matter physics. For example, a measure of \textit{first-order coherence} $\mathcal{D}=\sqrt{2\textrm{Tr}\rho^2-1}$ (for qubit subsystem) was
exploited to introduce the concept of accessible coherence \cite{Kagalwala2013} and it was recently demonstrated that $\mathcal{D}$ can be unveiled
from hidden coherence in partially coherent light \cite{Svozilik2015}. Moreover, Fang \textit{et al.} proposed a novel measure of the degree of symmetry
by virtue of group theoretical approach and showed that this quantity can effectively detect the phenomena of accidental degeneracy and
spontaneous symmetry breaking \cite{Fang2016}. In this work, we formulate two Frobenius-norm-based measures by further extending their
previous work and more importantly, illustrate the clear physical meanings of these quantifiers. Through the establishment of relationship
with other significant physical quantities (e.g., Brukner-Zeilinger information), we emphasize that the purity of quantum states is not only
a measure of mixedness, but also a \textit{basis-independent} quantification of quantum coherence contained in quantum states.

\section*{Results}
\subsection*{Coherence measure based on Frobenius norm}
In the famous work of Mandel and Wolf, they introduced an important measure depicting the \textit{degree of polarization}
for planar electromagnetic fields \cite{Mandel1995}
\begin{align}
\mathcal{P}_2=\sqrt{1-\frac{4\textrm{det}\Phi_2}{(\textrm{Tr}\Phi_2)^2}},
\end{align}
where $\Phi_2$ represents $2 \times 2$ equal-time coherence matrix at a given point $\vec{r}$
and $[\Phi_2]_{ij}=\langle E^\ast_i(\vec{r},t)E_j(\vec{r},t)\rangle$, $i,j=x,y$.
Note that this quantity characterizes the coherence between two mutually orthogonal components of the complex electric vector
$\mathbf{E}(t)=[E_x(t),E_y(t)]^\textrm{T}$. Be aware of the Hermiticity and non-negativity of the matrix $\Phi_2$, we can define a
physically valid quantum state $\rho=\Phi_2/\textrm{Tr}\Phi_2$. Since the Bloch sphere is analogous to the Poincar\'{e} sphere,
$\mathcal{P}_2$ can be rewritten as
\begin{align}
\mathcal{P}_2=\sqrt{1-4\textrm{det}\rho}=\sqrt{2\textrm{Tr}\rho^2-1}=\lambda_+-\lambda_-=|\vec{s}|,
\end{align}
where we have used the Bloch representation $\rho=(1+\vec{s}\cdot\vec{\sigma})/2$ and $\lambda_{\pm}=(1\pm|\vec{s}|)/2$ are
eigenvalues of $\rho$. It is apparent that $\mathcal{P}_2$ is equivalent to the degree of first-order coherence $\mathcal{D}$ adopted in
refs \cite{Kagalwala2013,Svozilik2015} and exactly equal to the length of the Bloch vector $\vec{s}$. In fact, in this case (e.g., for qubit system)
$\rho$ can also be recast as
\begin{align}
\rho=\mathcal{P}_2|\psi\rangle\langle\psi|+\frac{1-\mathcal{P}_2}{2}\mathbb{I}_{2},
\end{align}
with $|\psi\rangle\langle\psi|=(1+\vec{s_0}\cdot\vec{\sigma})/2$ and $\vec{s_0}=\vec{s}/|\vec{s}|$.

Furthermore, Set\"{a}l\"{a} \textit{et al.} moved a step forward and extended the above 2D formalism into the formulation of the
3D degree of polarization coherence for arbitrary electromagnetic fields \cite{Setala2002a,Setala2002b}, that is
\begin{align}
\mathcal{P}_3=\sqrt{\frac{3}{2}\left[\frac{\textrm{Tr}{\Phi^2_3}}{(\textrm{Tr}\Phi_3)^2}-\frac{1}{3}\right]},
\end{align}
where $\Phi_3$ denotes the $3\times3$ coherence matrix at a given point $\vec{r}$
and $[\Phi_3]_{ij}=\langle E^\ast_i(\vec{r},t)E_j(\vec{r},t)\rangle$, $i,j=x,y,z$.
Similarly, by redefining $\rho=\Phi_3/\textrm{Tr}\Phi_3$, we obtain
\begin{align}
\mathcal{P}_3=\sqrt{\frac{3}{2}\left[\textrm{Tr}{\rho^2}-\frac{1}{3}\right]},
\end{align}

Actually, we observe that this formalism can be unified and generalized to arbitrary dimensional systems, motivated by the work
of Luis \cite{Luis2007}, where a $4\times4$ density matrix was considered. The key idea is to quantify quantum coherence
from the geometrical point of view, that is, the degree of coherence between distinct variables of optical fields is
assessed as the distance between the coherence matrix and the identity matrix. Note that this geometrical approach
is also employed in quantum resource theories \cite{Horodecki2013,Brandao2015,Baumgratz2014,Vedral1997}. Here the identity matrix, or maximally
mixed state in the sense of quantum information, is identified as the fully incoherent and completely unpolarized light \cite{Mandel1995,Luis2007}.
Therefore, we can construct the following general measure
\begin{align}
\mathcal{C}(\rho)=\sqrt{\frac{d}{d-1}}\|\rho-\rho_\star\|_{\textrm{F}},
\label{QC-measure}
\end{align}
where $d$ is the dimension of the Hilbert space, $\rho_\star=\mathbb{I}_{d}/d$ denotes the maximally mixed state.
and the Frobenius norm is given by $\|A\|_{\textrm{F}}=\sqrt{\textrm{Tr}(A^\dagger A)}$.
Note that here the Frobenius norm is normalized to guarantee $\mathcal{C}(\rho)\in[0,1]$ (see Methods section for more details).
It is worth pointed out that this measure inherently possesses the following desirable properties:

(i) When $d=2,3,4$, the coherence measure (\ref{QC-measure}) reduces to the existing quantifiers for 2D, 3D and 4D cases,
where for 4D formalism, this measure was also termed as the \textit{depolarization index} \cite{Gil1986,Aiello2005};

(ii) No optimization is involved in $\mathcal{C}(\rho)$, which is in sharp contrast to the coherence measures proposed in ref. \cite{Baumgratz2014}, though
the optimization is rather straightforward there;

(iii) $\mathcal{C}(\rho)$ is independent of the specific representation of $\rho$, that is,
$\mathcal{C}(\rho)$ is \textit{basis-independent};

(iv) $\mathcal{C}(\rho)$ is unitarily invariant, since
$\mathcal{C}(U\rho U^\dagger)=\mathcal{C}(\rho)$ owing to the fact that the maximally mixed state $\rho_\star$
is \textit{the only} state that remains invariant under arbitrary unitary transformation. In fact, the property (iv)
is equivalent to (iii).

(v) $\mathcal{C}(\rho)$ has an analytical expression and explicit geometrical interpretation. As we shall see below,
$\mathcal{C}(\rho)$ also has a clear operational meaning.

To further elaborate on the characteristics of $\mathcal{C}(\rho)$, we first present several observations
where this coherence measure is involved.

\textit{Observation I:} Straightforward calculation shows that the square of $\mathcal{C}(\rho)$ is directly proportional to the
celebrated \textit{Brukner-Zeilinger invariant information} \cite{Brukner1999}
\begin{align}
\mathcal{C}^2(\rho)=\frac{d}{d-1}\left(\mathrm{Tr}\rho^2-\mathrm{Tr}\rho_\star^2\right)=\frac{d}{d-1}\left(\mathrm{Tr}\rho^2-\frac{1}{d}\right)
=\frac{d}{d-1}\mathcal{I}_{BZ}(\rho),
\end{align}
where the Brukner-Zeilinger information $\mathcal{I}_{BZ}$ is an operational notion defined as the sum of individual measures of information
over a complete set of mutually complementary observables (MCO) \cite{Brukner2001}
\begin{align}
\mathcal{I}_{BZ}(\rho)\doteq\sum_{i=1}^m\sum_{j=1}^d\left(\mathrm{Tr}(\Pi_{ij}\rho)-\frac{1}{d}\right)^2=\mathrm{Tr}\rho^2-\frac{1}{d}.
\end{align}
Here $i=1,\ldots,m$ and $j=1,\ldots,d$ label complementary observables and their eigenvectors, respectively. A significant feature
of $\mathcal{I}_{BZ}$ is its independence of the specific choice of the measured set of MCO and this independence is also equivalent to
the unitary invariance of $\mathcal{I}_{BZ}(U\rho U^\dagger)=\mathcal{I}_{BZ}(\rho)$, since the unitary transformations do not alter
the eigenvalues of $\rho$ and $\mathrm{Tr}\rho^2$ only relies on them. Therefore, it is evident that $\mathcal{C}(\rho)$ (and
$\mathcal{I}_{BZ}$) is a measure of \textit{basis-independent} coherence contained in quantum states and quantifies the intrinsic randomness
irrespective of the amount of entanglement between subsystems (if there are any). It is worth noting that the (global) purity
does not solely determine the entanglement of a multi-partite state and this fact led to the investigation of the so-called
\textit{maximally entangled mixed states} for a given degree of purity \cite{Munro2001,Wei2003}. Moreover, all pure states should always be
represented as a coherent superposition of a certain set of basis states rather than a classical mixing, and by our definition
$\mathcal{C}(\rho)=1$ for any pure state. This is reminiscent of the argument by von
Neumann that the entropy of all pure states is defined to be 0 as a kind of normalization \cite{Neumann1932}.
Finally, since the Brukner-Zeilinger invariant information was successfully utilized in quantum teleportation \cite{Lee2000},
state estimation \cite{Rehacek2002} and the violation of Bell's inequalities \cite{Brukner2002}, $\mathcal{C}(\rho)$,
as a renormalized version of $\mathcal{I}_{BZ}$, also plays an important role in all these quantum information tasks.

\textit{Observation II:} The coherence measure $\mathcal{C}(\rho)$ is associated with entropy production problem of doubly-stochastic
(bistochastic) quantum channels. By use of the quantum version of Kullback inequality (or a stronger Pinsker inequality \cite{Hiai1981})
\begin{align}
S(\rho\parallel\sigma)=\mathrm{Tr}(\rho\ln\rho-\rho\ln\sigma)\geq \frac{1}{2}\mathrm{Tr}^2|\rho-\sigma|\geq \frac{1}{2}\|\rho-\sigma\|_{\mathrm{F}}^2,
\label{Kullback}
\end{align}
Streater proved the following theorem \cite{Streater1985}:

\begin{theorem}
Let $\mathcal{H}$ be a Hilbert space with dim$\mathcal{H}=d<\infty$ and denote by $\mathcal{B}(\mathcal{H})_2$ the Hilbert space of operator on $\mathcal{H}$ with scalar product $\langle A,B\rangle=\mathrm{Tr}(A^\dagger B)$. If $T: \mathcal{B}(\mathcal{H})_2\rightarrow\mathcal{B}(\mathcal{H})_2$ is a bistochastic channel which is ergodic and has spectral gap
$\gamma\in[0,1)$ (e.g., up to a scalar multiple, the identity matrix $\mathbb{I}$ is the only fixed point of $T$ in $\mathcal{B}(\mathcal{H})_2$, and the spectrum
of the channel $T^\dagger\circ T$ is contained in the set $[0,1-\gamma]\cup\{1\}$), then for any density matrix $\rho$
\begin{align}
S(T(\rho))-S(\rho)\geq\frac{\gamma}{2}\|\rho-\rho_\star\|_{\textrm{F}}^2,
\end{align}
where $S(\rho)=-\mathrm{Tr}(\rho\ln\rho)$ is the von Neumann entropy of $\rho$.
\label{T1}
\end{theorem}

From Theorem \ref{T1}, we find that two separated terms contribute to the lower bound of the entropy production. The first term
$\|\rho-\rho_\star\|_{\textrm{F}}^2$ is just the unnormalized version of $\mathcal{C}(\rho)$, which characterizes the intrinsic
coherence independent of the choice of a specific basis, while the spectral gap $\gamma$ relies on representations of both $T$ and $\rho$.
To further elucidate the role of the bistochastic channel $T$, we consider the thermodynamics problem raised in ref. \cite{Kammerlander2016},
where the the bistochastic channel was taken to be a projection measurement $\rho\rightarrow\Delta(\rho):=\sum_k\Pi_k\rho\Pi_k$.
If $\{|l\rangle\}_l$ is the eigenbasis of $\rho$ and $\{\Pi_k=|\phi_k\rangle\langle\phi_k|\}_k$ constitute another orthogonal basis,
then $\gamma$ is the second smallest eigenvalue of the matrix $\mathbb{I}_d-M^\mathrm{T}M$, where $M$ is a bistochastic matrix
with the entries $M_{kl}=\langle l|\Pi_k| l\rangle=|\langle\phi_k| l\rangle|^2$. It is easy to show that $\gamma$ is zero if
$\{|l\rangle\}_l$ and $\{|\phi_k\rangle\}_k$ are the same, while $\gamma=1$ if they are mutually unbiased.
A deeper insight can be gained by noticing that once quantum thermodynamics is viewed as a resource theory, it can be cast as
a \textit{hybrid} of the resource theory of purity and the resource theory of asymmetry \cite{Horodecki2013}.
In particular, if $\{\Pi_k\}_k$ are the eigenvectors of an observable $L$, $\gamma$ (e.g., second smallest eigenvalue of $\mathbb{I}_d-M^\mathrm{T}M$)
is a quantifier to characterize to what extent $\rho$ fails to commute with $L$ and $\gamma\neq0$ if and only if $\rho$
has some coherence over the eigenspaces of $L$. This line of thought coincides with that of Marvian \textit{et al.} \cite{Marvian2014b}.

\textit{Observation III:} For a square $d$-dimensional matrix $A$ and $p\in[1,+\infty)$, the $l_p$ norms (or vector norms)
and Schatten-$p$ norms are defined as \cite{Watrous2011}
\begin{align}
\|A\|_{l_p}:=\left(\sum_{i,j=1}^d|A_{ij}|^p\right)^{1/p}, \,
\|A\|_{p}:=\left(\mathrm{Tr}|A|^p\right)^{1/p}=\left(\sum_{i=1}^r\sigma_i^p\right)^{1/p}=\|\vec{s}(A)\|_{p},
\end{align}
where $|A|=\sqrt{A^\dagger A}$, $r=rank(A)$ and $\vec{s}(A)$ is used to refer to the vector of nonzero singular values
$\vec{s}(A)=(\sigma_1,\ldots,\sigma_r)$. In addition to the monotonicity $\|A\|_1\geq\|A\|_p\geq\|A\|_q\geq\|A\|_\infty$
for $1\leq p\leq q\leq\infty$, we can prove the following inequality by using H\"{o}lder's inequality \cite{Bhatia1997}
\begin{align}
\|A\|_{p}\leq d^{\frac{1}{p}-\frac{1}{q}}\|A\|_{q},\quad \forall \, q\geq p\geq 1.
\label{inequality1}
\end{align}
Note that the monotonicity and the above inequality (\ref{inequality1}) is also satisfied by $l_p$ norms, and
the Schatten-$p$ norms of $A$ coincides with the ordinary (vector) $l_p$ norms of $\vec{s}(A)$. In particular,
the Frobenius norm of $A$ coincides with the corresponding $l_2$ norm. Therefore, $\mathcal{C}(\rho)$
gives a upper bound of $l_1$ norm coherence and trace-norm coherence discussed in ref. \cite{Baumgratz2014} and \cite{Rana2016} respectively
\begin{align}
C_{l_1}(\rho)&=\min_{\delta\in\mathcal{I}}\|\rho-\delta\|_{l_1}\leq\|\rho-\rho_\star\|_{l_1}\leq\sqrt{d^2}\|\rho-\rho_\star\|_{l_2}=\sqrt{d(d-1)}\mathcal{C}(\rho),\label{inequality2}\\
C_{\mathrm{tr}}(\rho)&=\min_{\delta\in\mathcal{I}}\|\rho-\delta\|_{\mathrm{Tr}}\leq\|\rho-\rho_\star\|_{\mathrm{Tr}}\leq\sqrt{d}\|\rho-\rho_\star\|_{\mathrm{F}}
=\sqrt{d-1}\mathcal{C}(\rho).\label{inequality3}
\end{align}
where $\mathcal{I}$ denotes the set of diagonal states defined in a pre-fixed basis. It is not surprising that
the upper bounds (\ref{inequality2}) and (\ref{inequality3}) are not tight in general though $C_{l_1}(\rho)$ and $C_{\mathrm{tr}}(\rho)$
indeed achieve their maximum values with the closest state $\delta=\rho_\star$.
In fact, $C_{l_1}^{\mathrm{max}}(\rho)=d-1$ and $C_{\mathrm{tr}}^{\mathrm{max}}(\rho)=2(1-1/d)$, which is satisfied by
the maximally coherent state \cite{Baumgratz2014}.

On the other hand, it is intuitive to assume that the intrinsic coherence measure $\mathcal{C}(\rho)$ may have a natural
relationship with the resource theory of purity \cite{Horodecki2003}, due to the fact that $\mathcal{C}(\rho)$ is proportional to
$\mathcal{I}_{BZ}(\rho)=\mathrm{Tr}\rho^2-1/d$. Actually, $\mathcal{I}_{BZ}(\rho)$ provide a lower bound for the so-called
\textit{unique measure of information} introduced in ref. \cite{Horodecki2002}. By virtue of the quantum Kullback inequality (\ref{Kullback}),
we have
\begin{align}
\mathcal{I}_{BZ}(\rho)/2=\|\rho-\rho_\star\|_{\textrm{F}}^2/2\leq S(\rho\|\rho_\star)=\log_2d-s(\rho),
\end{align}
where $I(\rho)=\log_2d-s(\rho)$ is a unique measure for information, in the sense that $I(\rho)=\log_2d-s(\rho)$ is equal to
the optimal transition rate for mixed states to one qubit pure state $\pi$ (i.e., $\pi=|0\rangle\langle0|$) \cite{Horodecki2013,Horodecki2003}.
However, since the Kullback inequality is not very sharp, this lower bound is rather loose, especially for the states with high purity.

\subsection*{Cohering (Purifying) power of quantum channels}
With the framework of Baumgratz \textit{et al.} \cite{Baumgratz2014}, we proposed the notion of cohering power of quantum channels
(i.e., completely positive and trace preserving maps)\cite{Yao2015}, which has been further explored in ref. \cite{Mani2015}.
In the context of this work, we can introduce a similar quantity to characterize the cohering (purifying) capacity of quantum channels
\begin{align}
\mathcal{C}(\mathcal{E})=\sqrt{\frac{d}{d-1}}\|\mathcal{E}(\rho_\star)-\rho_\star\|_{\textrm{F}}.
\label{CC-measure}
\end{align}
Complete positivity of $\mathcal{E}(\cdot)$ enables a Kraus decomposition
\begin{align}
\mathcal{E}(\rho)=\sum_\mu K_\mu\rho K_\mu^\dagger, \quad \sum_\mu K_\mu^\dagger K_\mu=\mathbb{I}_d.
\end{align}
where the second relation stems from the trace preserving property. Thus, the cohering power (\ref{CC-measure})
can be rewritten as
\begin{align}
\mathcal{C}(\mathcal{E})=\sqrt{\frac{1}{d(d-1)}\mathrm{Tr}\left[\left(\sum_\mu\left[K_\mu,K_\mu^\dagger\right]\right)^2\right]}.
\end{align}
It is observed that the operator $\sum_\mu[K_\mu,K_\mu^\dagger]$ is Hermitian and traceless. Moreover, $\mathcal{C}(\mathcal{E})=0$
if and only if $\mathcal{E}(\cdot)$ is a unital quantum channel, i.e., doubly-stochastic (bistochastic) completely positive map
or $\mathcal{E}(\rho_\star)=\rho_\star$ \cite{Mendl2009}.
Therefore, the cohering capacity $\mathcal{C}(\mathcal{E})$ can be recognized as a measure of \textit{nonunitality} of quantum channels.
Note that the nonunital channels also play a crucial role in the local creation of quantum correlations \cite{Streltsov2011}.
Within the scope of resource theories \cite{Horodecki2013,Brandao2015}, we \textit{may} make a correspondence in our context:
the only free state is the maximally mixed state $\rho_\star$, the free operations are unital channels and
the resources are states with $\mathcal{C}(\rho)>0$. The justification of such a correspondence relies on the monotonicity
of the function $\mathcal{C}(\rho)$ under unital channels. The following theorem proves that it is the case.

\begin{theorem}
For any unital (i.e, bistochastic) quantum channel, it holds that $\mathcal{C}(\mathcal{E}(\rho))\leq\mathcal{C}(\rho)$.
\label{T2}
\end{theorem}

\begin{proof}
Since $\mathcal{C}(\rho)\propto\mathrm{Tr}\rho^2-\mathrm{Tr}\rho_\star^2$, it is sufficient to prove
$\mathrm{Tr}(\mathcal{E}^2(\rho))\leq\mathrm{Tr}\rho^2$. Consider the Kraus decomposition of $\mathcal{E}(\cdot)=\sum_\mu K_\mu(\cdot) K_\mu^\dagger$ and we have
\begin{align}
\mathrm{Tr}(\mathcal{E}^2(\rho))&=\mathrm{Tr}(\sum_{\mu,\nu}K_\mu\rho K_\mu^\dagger K_\nu\rho K_\nu^\dagger)  \nonumber\\
&=\mathrm{Tr}(\sum_{\mu,\nu}\rho K_\mu^\dagger K_\nu\rho K_\nu^\dagger K_\mu)  \nonumber\\
&=\mathrm{Tr}(\sum_{\mu,\nu}\rho A_{\mu\nu}\rho A_{\mu\nu}^\dagger),
\end{align}
where we have defined $A_{\mu\nu}=K_\mu^\dagger K_\nu$. It is interesting to see that $\{A_{\mu\nu}\}$ also constitute a unital quantum channel
$A(\cdot)=\sum_{\mu,\nu}A_{\mu\nu}(\cdot) A_{\mu\nu}^\dagger$ due to the identities
\begin{align}
\sum_\mu A_{\mu\nu}^\dagger A_{\mu\nu}=\sum_\mu K_\nu^\dagger K_\mu  K_\mu^\dagger K_\nu =K_\nu^\dagger K_\nu, \\
\sum_\nu A_{\mu\nu}A_{\mu\nu}^\dagger=\sum_\nu K_\mu^\dagger K_\nu K_\nu^\dagger K_\mu=K_\mu^\dagger K_\mu,
\end{align}
where the unitality of $\mathcal{E}(\cdot)$ is used. With the aid of Chauchy-Schwarz inequality for operators
$|\mathrm{Tr}(A^\dagger B)|\leq \|A\|_p\|B\|_q$ for $p^{-1}+q^{-1}=1$ (e.g., $p=q=2$ in our case), we finally obtain
\begin{align}
\mathrm{Tr}(\mathcal{E}^2(\rho))&=\sum_{\mu,\nu}\mathrm{Tr}(\rho A_{\mu\nu}\rho A_{\mu\nu}^\dagger)
\leq\sum_{\mu,\nu}\sqrt{\mathrm{Tr}(\rho A_{\mu\nu} A_{\mu\nu}^\dagger\rho)}\sqrt{\mathrm{Tr}(A_{\mu\nu}\rho\rho A_{\mu\nu}^\dagger)}  \nonumber\\
&\leq\frac{1}{2}\sum_{\mu,\nu}\left[\mathrm{Tr}(\rho A_{\mu\nu} A_{\mu\nu}^\dagger\rho)+\mathrm{Tr}(A_{\mu\nu}\rho\rho A_{\mu\nu}^\dagger)\right]  \nonumber\\
&=\mathrm{Tr}\rho^2.
\end{align}
where the unitality property of $A(\cdot)$ is applied.
\end{proof}
We have attempted to check whether the unitality of quantum channels is also a necessary condition for the monotonicity of $\mathcal{C}(\rho)$.
Note that for finite-dimensional Hilbert spaces, a dynamical semigroup is strictly purity-decreasing if and only if the Lindblad generator is unital \cite{Lidar2006}.
However, actually the Kraus or operator-sum representation is more general than the master equation approach \cite{Nielsen2000}. Therefore,
the necessity of unitality of operator-sum representation for strictly-decreasing purity is still an important unsolved problem.
See Method section for more information.

\subsection*{Asymmetry measure based on Frobenius norm}
In realistic physical scenarios, there are some states which are invariant under a set of symmetric operations or
a given symmetric group. This fact tells us that a specific group of symmetric transformations assigns a
specific type of symmetry to the underlying systems. A state $\rho$ is called symmetric (or $G$-invariant) with respect
to a finite or compact Lie group $G$ if it satisfies \cite{Marvian2013,Marvian2014b,Vaccaro2008}
\begin{align}
\mathcal{U}_g(\rho)=U(g)\rho U^\dagger(g)=\rho, \quad \forall \, g\in G,
\end{align}
where $U(g)$ denotes a unitary representation corresponding to the group element $g\in G$.
We notice that this definition of symmetry is equivalent to other two criterions
\begin{align}
\mathcal{G}_G(\rho)=\frac{1}{|G|}\sum_{g\in G} U(g)\rho U^\dagger(g)=\rho \stackrel{(i)}{\Longleftrightarrow}  \mathcal{U}_g(\rho)=\rho, \, \forall \, g\in G
\stackrel{(ii)}{\Longleftrightarrow} [U(g),\rho]=0,\, \forall \, g\in G,
\end{align}
where the summation in $\mathcal{G}_G(\rho)$ will be replaced by integral over $dg$ when a compact Lie group is considered.
The criterion (i) induced a nature entropic measure of the asymmetry of $\rho$ with respect to $G$ \cite{Marvian2014b,Vaccaro2008},
(i.e., $A_G(\rho)=S(\mathcal{G}_G(\rho))-S(\rho)$), which was also proved to be a measure of the quality of a quantum reference frame \cite{Gour2009}.
Based on the criterion (ii), it is realized that the commutator $[U(g),\rho]$ characterizes to what extent $\rho$ is asymmetric with respect to $G$
and some matrix norm of $[U(g),\rho]$ should be a reasonable measure of $G$-asymmetry \cite{Marvian2014b}.

Along this line of thought, a new quantifier of asymmetry has recently been proposed by Sun's group, which is defined as an average
of the \textit{fidelity deviations} of Hamiltonian $H$ (or quantum state $\rho$) over a specific symmetric group $g\in G$ \cite{Fang2016}
\begin{align}
\mathcal{A}(G,H)=\frac{1}{4\|H_0\|_\mathrm{F}^2}\int_g\left\|[U(g),H]\right\|_\mathrm{F}^2dg,
\end{align}
where $H_0=H-\mathrm{Tr}H/d$ is a re-biased version of $H$. $\mathcal{A}(G,H)$ is mathematically tractable and also has desirable
properties (e.g., basis-independence), mainly owing to the geometric feature of the Frobenius norm \cite{Bhatia1997}. More importantly,
it is has been demonstrated that $\mathcal{A}(G,H)$ can effectively detect some significant phenomena in condensed matter physics,
such as accidental degeneracy and spontaneous symmetry breaking \cite{Fang2016}. Meanwhile, an analogous measure can be defined
since we are more concerned with how to quantify the asymmetry (symmetry) of quantum states
\begin{align}
\mathcal{A}(G,\rho)&=\frac{1}{4\|\rho\|_\mathrm{F}^2}\int_g\left\|[U(g),\rho]\right\|_\mathrm{F}^2dg,\\
\mathcal{S}(G,\rho)&=1-\mathcal{A}(G,\rho)=\frac{1}{4\|\rho\|_\mathrm{F}^2}\int_g\left\|\{U(g),\rho\}\right\|_\mathrm{F}^2dg.
\end{align}

With this formulation, we can investigate the asymmetry of any quantum state with respect to a specific symmetric group.
Remarkably, while seemingly unrelated, there actually exists a close connection between $\mathcal{A}(G,\rho)$ and $\mathcal{C}(\rho)$.
We first focus on the single qubit system $\rho=(1+\vec{s}\cdot\vec{\sigma})/2$, where $\vec{s}$ is the Bloch vector and
can be parameterized as $\vec{s}=|\vec{s}|(\sin\theta_0\cos\phi_0,\sin\theta_0\sin\phi_0,\cos\theta_0)$
with polar and azimuthal angles $(\theta_0,\phi_0)$. Since the group of all linear unitary operations
over a single qubit is equivalent to a $\mathbb{SU}(2)$ algebra, here what interests us most is to investigate
the performance of $\mathcal{A}(G,\rho)$ under the $\mathbb{SU}(2)$ group. The group element of $\mathbb{SU}(2)$
can be expressed as
\begin{align}
R(\omega,\theta,\phi)=\exp\left(-i\frac{\omega}{2}\hat{n}\cdot\vec{\sigma}\right)
=\left(\begin{array}{cc}
\cos\frac{\omega}{2}-i\sin\frac{\omega}{2} \cos\theta& -i\mathrm{e}^{-i\phi}\sin\frac{\omega}{2}\sin\theta \\
-i\mathrm{e}^{i\phi}\sin\frac{\omega}{2}\sin\theta & \cos\frac{\omega}{2}+i\sin\frac{\omega}{2}\cos\theta
\end{array}\right),
\end{align}
where $\hat{n}$ is a real unit vector and for simplicity we denote the triple as a vector $\vec{\omega}=\{\omega,\theta,\phi\}$.
With the identity $\|[R(\vec{\omega}),\rho]\|_\mathrm{F}^2=2\mathrm{Tr}\rho^2-2\mathrm{Tr}[\rho R(\vec{\omega})\rho R(\vec{\omega})^\dagger]$
and $\|\rho\|_\mathrm{F}^2=\mathrm{Tr}\rho^2$, the quantity $\mathcal{A}(G,\rho)$ for $\mathbb{SU}(2)$ group can be explicitly calculated
\begin{align}
\mathcal{A}(\mathbb{SU}(2),\rho)&=\frac{1}{4\|\rho\|_\mathrm{F}^2}\int_{\vec{\omega}}\|[R(\vec{\omega}),\rho]\|_\mathrm{F}^2d\vec{\omega}
=\frac{1}{2}-\frac{1}{2\mathrm{Tr}\rho^2}\int_{\vec{\omega}}\mathrm{Tr}\left[\rho R(\vec{\omega})\rho R(\vec{\omega})^\dagger\right] d\vec{\omega} \nonumber\\
&=\frac{1}{2}-\frac{1}{2\mathrm{Tr}\rho^2}\frac{1}{4\pi^2}\int_{-\pi}^\pi d\phi \int_0^\pi\sin\theta d\theta
\int_0^{2\pi}\sin^2\frac{\omega}{2} F(\omega,\theta,\phi) d\omega\nonumber\\
&=\frac{1}{2}-\frac{1}{4\mathrm{Tr}\rho^2}=\frac{1}{2}-\frac{1}{2(1+|\vec{s}|^2)} \label{SU2},
\end{align}
where we have used the integral formula for a functional $F(\vec{\omega})$ over the $\mathbb{SU}(2)$ group \cite{Gilmore1974}
\begin{align}
\int d\vec{\omega}F(\vec{\omega})=\frac{1}{4\pi^2}\int_{-\pi}^\pi d\phi \int_0^\pi\sin\theta d\theta
\int_0^{2\pi}\sin^2\frac{\omega}{2} F(\omega,\theta,\phi) d\omega,
\end{align}
and in our case the functional $F(\omega,\theta,\phi)=\mathrm{Tr}[\rho R(\vec{\omega})\rho R(\vec{\omega})^\dagger]$ can be represented as
\begin{align}
F(\omega,\theta,\phi)=\frac{1}{16}\Big\{8+3|\vec{s}|^2+|\vec{s}|^2\left[5\cos\omega+2\sin^2\frac{\omega}{2}(\cos2\theta_0+\cos2\theta+3\cos2\theta\cos2\theta_0)\right] \nonumber\\
+8|\vec{s}|^2\sin^2\frac{\omega}{2}\left[\cos2(\phi-\phi_0)\sin^2\theta\sin^2\theta_0+\cos(\phi-\phi_0)\sin2\theta\sin2\theta_0\right]\Big\}
\end{align}

As expected, the result (\ref{SU2}) indicates that $\mathcal{A}(\mathbb{SU}(2),\rho)$ is a basis-independent quantity, that is,
it does not rely on the initial orientation of $\vec{s}$, although $F(\omega,\theta,\phi)$ explicitly depend on $\theta_0$ and $\phi_0$.
Moreover, when $|\vec{s}|$=0, the initial state $\rho$ is the maximally mixed state $\rho_\star$ and meanwhile
$\mathcal{A}(\mathbb{SU}(2),\rho_\star)=0$. This observation coincides with the fact that for any symmetry and for any
representation of the symmetry, the completely mixed state $\rho_\star$ is invariant under all symmetry
transformations \cite{Marvian2013}. Intriguingly, we build up a relationship between $\mathcal{A}(\mathbb{SU}(2),\rho)$
and $\mathcal{C}(\rho)$ for arbitrary single-qubit system
\begin{align}
\mathcal{A}(\mathbb{SU}(2),\rho)=\frac{1}{2}-\frac{1}{2[1+\mathcal{C}^2(\rho)]}.
\end{align}
Apparently, $\mathcal{A}(\mathbb{SU}(2),\rho)$ is a monotonic increasing function of $\mathcal{C}(\rho)$, which implies that
more intrinsic coherence signifies more asymmetry under the $\mathbb{SU}(2)$ group. Such a simple but novel relation
motivates us to inspect $\mathcal{A}(G,\rho)$ for $N$-qubit states by invoking the tensor product structure
of Hilbert space.

An arbitrary $N$-qubit (pure or mixed) states can be written as
\begin{align}
\rho=\frac{1}{2^N}\sum_{x_1,\ldots,x_N=0}^3T_{x_1\cdots x_N} \sigma_{x_1}^1\otimes\cdots\otimes\sigma_{x_N}^N,
\end{align}
where $\sigma_{x_j}^j$ ($x_j=1,2,3$) are the standard Pauli matrices in the Hilbert space of
qubit $j$ and $\sigma_{0}^j$ is the corresponding identity operator. The set of real coefficients $T_{x_1\cdots x_N}$
with $x_j=1,2,3$ constitutes the so-called correlation matrix $T$. To gain further insight into the property of
$\mathcal{A}(G,\rho)$, we consider the symmetry group $G=\mathbb{SU}(2)\otimes\cdots\otimes\mathbb{SU}(2)$,
where the group element is $R_{\mathrm{ind}}=R_1\otimes\cdots\otimes R_N$ and $R_j\in\mathbb{SU}(2)$ are \textit{independent} local
unitary operations acting on qubit $j$. Similar to the single-qubit case (\ref{SU2}), the key procedure is to calculate
the average of $\mathrm{Tr}[\rho R_{\mathrm{ind}}\rho R_{\mathrm{ind}}^\dagger]$ over $\mathbb{SU}(2)^{\otimes N}$.
Due to the tensor product structure of N-qubit states and the trace formula $\mathrm{Tr}(A\otimes B)=\mathrm{Tr}(A)\mathrm{Tr}(B)$,
we have the relation
\begin{align}
\mathrm{Tr}[\rho R_{\mathrm{ind}}\rho R_{\mathrm{ind}}^\dagger]=\frac{1}{4^N}\sum_{x_1,\ldots,x_N=0}^3\sum_{y_1,\ldots,y_N=0}^3
T_{x_1\cdots x_N}T_{y_1\cdots y_N}\prod_{j=1}^N\mathrm{Tr}(\sigma_{x_j}^jR_j\sigma_{y_j}^jR_j^\dagger).
\end{align}
Therefore, the problem reduces to the evaluation of the integral $\int_{\vec{\omega}}\mathrm{Tr}(\sigma_{x_j}^jR_j\sigma_{y_j}^jR_j^\dagger)d\vec{\omega}$.
In fact, we can prove that only when $x_j=y_j=0$, this integral is nonzero, namely (see Methods section)
\begin{align}
\int_{\vec{\omega}}\mathrm{Tr}(\sigma_{x_j}^jR_j\sigma_{y_j}^jR_j^\dagger)d\vec{\omega}=
\left\{\begin{array}{cc}
2, &  \mathrm{if} \quad x_j=y_j=0\\
0, &  \mathrm{otherwise}
\end{array}\right.
\end{align}
Finally, we obtain the asymmetry of $N$-qubit states under $\mathbb{SU}(2)^{\otimes N}$ in a general form
\begin{align}
\mathcal{A}(\mathbb{SU}(2)^{\otimes N},\rho)&=\frac{1}{2}-\frac{1}{2\mathrm{Tr}\rho^2}
\frac{1}{4^N}\sum_{x_1,\ldots,x_N=0}^3\sum_{y_1,\ldots,y_N=0}^3T_{x_1\cdots x_N}T_{y_1\cdots y_N}
\prod_{j=1}^N\int_{\vec{\omega}}\mathrm{Tr}(\sigma_{x_j}^jR_j\sigma_{y_j}^jR_j^\dagger)d\vec{\omega} \nonumber\\
&=\frac{1}{2}-\frac{1}{2\mathrm{Tr}\rho^2}\frac{1}{4^N}T_{0\cdots 0}T_{0\cdots 0}2^N=\frac{1}{2}-\frac{1}{2^{N+1}\mathrm{Tr}\rho^2} \nonumber\\
&=\frac{1}{2}-\frac{1}{2+2(2^{N}-1)\mathcal{C}^2(\rho)}=\frac{1}{2}-\frac{1}{2+2^{N+1}\mathcal{I}_{BZ}(\rho)},
\end{align}
where the coefficient $T_{0\cdots 0}=\mathrm{Tr}\rho=1$. Obviously, this result reduces to the single-qubit case
for $N=1$ and the maximally mixed state $\rho_\star=\mathbb{I}_{2^N}/2^N$ still corresponds to the minimum
value of $\mathcal{A}(G,\rho)$ ($\mathcal{A}(G,\rho_\star)=0$). In addition to the unambiguous and
monotonous relation between $\mathcal{A}(\mathbb{SU}(2)^{\otimes N},\rho)$ and $\mathcal{C}(\rho)$ (hence $\mathcal{I}_{\mathrm{BZ}}(\rho)$),
it is worthwhile to note that $\mathcal{A}(\mathbb{SU}(2)^{\otimes N},\rho)$ is independent of
the correlation matrix $T$, more precisely, independent of the internal quantum correlations (e.g., entanglement) between subsystems.

Intuitively, one might be tempted to conjecture that \textit{any} local symmetric transformation group $G$ will lead to
a \textit{correlation-independent} measure of asymmetry. However, this is not the case.
To illustrate this, let us take a closer look at local unitary transformations.
Other than the independent local unitary transformations considered above, their \textit{collective}
counterparts also plays a crucial role in quantum information and computation \cite{Bourennane2004}.
In this circumstance, the element of symmetry group is the tensor product of $N$ \textit{identical} unitary transformations $R(\vec{\omega})$,
that is, $R_{\mathrm{col}}=R\otimes\cdots\otimes R$.
Take two-qubit states for example and we obtain (see Methods section for more details)
\begin{align}
\mathcal{A}(R\otimes R,\rho)&=\frac{1}{2}-\frac{1}{2\mathrm{Tr}\rho^2}
\frac{1}{16}\sum_{x_1,x_2=0}^3\sum_{y_1,y_2=0}^3T_{x_1x_2}T_{y_1y_2}
\int_{\vec{\omega}}\mathrm{Tr}(\sigma_{x_1}^1R\sigma_{y_1}^1R^\dagger)\mathrm{Tr}(\sigma_{x_2}^2R\sigma_{y_2}^2R^\dagger)d\vec{\omega} \nonumber\\
&=\frac{1}{2}-\frac{1}{2\mathrm{Tr}\rho^2}\frac{1}{16}\left(4+\frac{4}{3}\sum_{i,j=1}^3T_{ii}T_{jj}\right)
=\frac{1}{2}-\frac{1}{8\mathrm{Tr}\rho^2}\left(1+\frac{1}{3}\sum_{i,j=1}^3T_{ii}T_{jj}\right).
\end{align}
It is evident that the collective effect exhibits a significant influence on the asymmetry measure
and $\mathcal{A}(R\otimes R,\rho)$ indeed relies on the diagonal entries of the correlation matrix $T$,
which are basis-dependent quantities.

\section*{Discussion}
In this work, we first formulate a coherence measure $\mathcal{C}(\rho)$ based on the Frobenius norm,
which is defined from a geometric perspective. Remarkably, we have demonstrated that this quantity
is not only a measure of mixedness, but also an intrinsic (\textit{basis-independent}) quantification
of quantum coherence contained in quantum states, which can also be viewed as a normalized version of Brukner-Zeilinger information \cite{Brukner1999}.
To further illustrate this point, a comparison can be
made between $\mathcal{C}(\rho)$ and the coherence measures proposed by Baumgratz \textit{et al.} \cite{Baumgratz2014}.
For example, for a single-qubit state $\rho=(1+\vec{s}\cdot\vec{\sigma})/2$, those quantifiers
are listed as follows
\begin{align}
\mathcal{C}(\rho)=&|\vec{s}|=\sqrt{s_1^2+s_2^2+s_3^2}, \  C_{R}(\rho)=h\left(\frac{1+s_3}{2}\right)-h\left(\frac{1+|\vec{s}|}{2}\right), \\
C_{l_1}(\rho)&=C_{\mathrm{tr}}(\rho)=\frac{1}{2}(|s_1-is_2|+|s_1-is_2|)=\sqrt{s_1^2+s_2^2}=\sqrt{|\vec{s}|^2-s_3^2},
\end{align}
where $h(x)=-x\log x-(1-x)\log(1-x)$ is the binary entropy function and $C_R(\rho)$ is the relative entropy of coherence.
From the above expressions, it is evident that $C_{R}(\rho)$, $C_{l_1}(\rho)$ and $C_{\mathrm{tr}}(\rho)$
are all basis-dependent measures, since they all depend on $s_3$, which is not a unitarily-invariant quantity.
Moreover, the participation of $s_3$ precisely reflects the choice of the pre-fixed basis. In fact, we can
present an alternative formulation of $\mathcal{C}(\rho)$ based on the eigenvalues $\lambda_j$ of $\rho$,
by noticing that $\lambda_j-1/d$ is the corresponding eigenvalues of $\rho-\rho_\star$, which is a Hermitian and
traceless operator
\begin{align}
\mathcal{C}(\rho)=\sqrt{\frac{d}{d-1}\sum_{j=1}^d(\lambda_j-\frac{1}{d})^2}=\sqrt{\frac{1}{2(d-1)}\sum_{j,k=1}^d(\lambda_j-\lambda_k)^2}.
\end{align}
Note that $\mathcal{C}(\rho)=1$ for any pure state. This result is just a consequence of the basis-dependent property of $\mathcal{C}(\rho)$:
pure states can only be represented as a coherent superposition in \textit{any} pre-fixed basis (without including itself as a base vector), and more precisely,
any pure state can always be connected by a unitary transformation within a given dimension.

Analogous to previous works \cite{Yao2015,Mani2015}, the cohering power $\mathcal{C}(\mathcal{E})$ of quantum channels is introduced in our context.
It is demonstrated that $\mathcal{C}(\mathcal{E})=0$ if and only if $\mathcal{E}(\cdot)$ is a unital quantum channel, i.e.,
doubly-stochastic (bistochastic) completely positive map.
Therefore, the cohering capacity $\mathcal{C}(\mathcal{E})$ can be recognized as a measure of \textit{nonunitality} of quantum channels.
For instance, for some common types of decoherence processes (such as depolarizing and phase damping channels) $\mathcal{C}(\mathcal{E})=0$
since they are unital quantum channels. However, for dissipative channels (e.g., amplitude damping channel)
$\mathcal{C}(\mathcal{E})$ is exactly equal to the damping parameter $\gamma$, which can be thought of as the
probability of losing a photon (here we adopt the notations in ref. \cite{Nielsen2000}). Intrudingly, the
spontaneous emission process is one of the strongest cohering channels, due to the fact that
it maps any state into a certain pure state (e.g., the ground state of a system). Based on the commutator
between quantum state and the elements of a specific symmetric group $G$, we formulate an asymmetry measure $\mathcal{A}(G,\rho)$
by integrating the Frobenius norm of the commutator over this group. Two distinct situations are considered: local
\textit{independent} and \textit{collective} $\mathbb{SU}(2)$ transformations. For $N$-qubit quantum states,
$\mathcal{A}(G,\rho)$ only relies on the purity and is actually equivalent to $\mathcal{C}(\rho)$ (or $I_{BZ}$)
under local independent $\mathbb{SU}(2)$ transformations. However, for local collective $\mathbb{SU}(2)$ transformations,
$\mathcal{A}(G,\rho)$ also depends on the quantum correlations between subsystems. One important open question
is to find a closed expression of $\mathcal{A}(G,\rho)$ for arbitrary $N$-qubit states (especially $N>3$) in this case.

\section*{Methods}
\subsection*{Normalization of $\mathcal{C}(\rho)$ and $\mathcal{C}(\mathcal{E})$}
Here we only focus on $\mathcal{C}(\mathcal{E})$, because $\mathcal{E}(\cdot)$ is a
completely positive and trace preserving map and hence $\mathcal{E}(\rho)$
is also a valid density matrix. Trace-preserving property renders
$\mathrm{Tr}(\mathcal{E}(\rho_\star))=\mathrm{Tr}(\rho_\star)=1$, and we obtain
\begin{align}
\|\mathcal{E}(\rho_\star)-\rho_\star\|_\mathrm{F}^2=\|\mathcal{E}(\rho_\star)\|_\mathrm{F}^2-\frac{1}{d}.
\end{align}
Note that for $p\in[1,+\infty)$ and an arbitrary \textit{non-negative} vector $\vec{x}$, the following inequality holds for vector $p$-norms \cite{Roga2013}
\begin{align}
\|\vec{x}\|_p\leq\|\vec{x}\|_1^{1/p}\|\vec{x}\|_\infty^{(p-1)/p}.
\end{align}
Since the Schatten-$p$ norms of an operator $A$ is equal to its vector $p$-norm of its singular values (e.g., $\|A\|_{p}=\|\vec{s}(A)\|_{p}$),
the above inequality also holds for Schatten-$p$ norms. Therefore, we have
\begin{align}
\|\mathcal{E}(\rho_\star)-\rho_\star\|_\mathrm{F}^2\leq\|\mathcal{E}(\rho_\star)\|_1\|\mathcal{E}(\rho_\star)\|_\infty-\frac{1}{d}
=\|\mathcal{E}(\rho_\star)\|_\infty-\frac{1}{d}\leq1-\frac{1}{d},
\end{align}
where we have used $\|\mathcal{E}(\rho_\star)\|_1=\mathrm{Tr}[\mathcal{E}(\rho_\star)]=1$ due to the positivity of $\mathcal{E}(\cdot)$ and the monotonicity
$\|\mathcal{E}(\rho_\star)\|_\infty\leq\|\mathcal{E}(\rho_\star)\|_1$. It is quite clear that the equality is satisfied only if
the quantum channel $\mathcal{E}(\cdot)$ maps the maximally mixed state $\rho_\star$ to a pure state, which only possesses intrinsic
(e.g., basis-independent) coherence in our context.

\subsection*{Necessity of unital channels}
In ref \cite{Lidar2006}, it is proved that a dynamical semigroup is strictly purity-decreasing if and only if the Lindblad generator is unital.
Here the quantum dynamic semigroup (in the Schr\"{o}dinger picture) is a family of one-parameter linear trace-preserving maps $\{\mathcal{E}_\tau,\tau\geq0\}$
satisfying an additional property except for complete positivity, namely \cite{Alicki2007}
\begin{align}
\mathcal{E}_{\tau_1}\mathcal{E}_{\tau_2}=\mathcal{E}_{\tau_1+\tau_2},
\end{align}
which is called semigroup condition or Markov property. However, we recall that a linear map $\mathcal{E}$ is completely
positive \textit{if and only if} it admits a Kraus representation. Thus in our context we are dealing with a more general class of dynamical maps.
Note that Nielsen and Chuang stated in their book that a quantum process described in terms of an operator-sum
representation cannot necessarily be written down as a master equation, especially for non-Markovian dynamics \cite{Nielsen2000}.
In particular, it is worth pointing out that the technique used in ref. \cite{Lidar2006} does not work for our purpose.
Therefore, the necessity of unitality of operator-sum representation for strictly-decreasing purity is still an important unsolved problem.

Moreover, though we present an elegant proof of the sufficient part, we can reach this conclusion from a different perspective.
For $\alpha\in(0,1)\cup(1,+\infty)$, the quantum $\alpha$-divergence \cite{Rastegin2016}
\begin{align}
D_\alpha(\rho\|\sigma)=\frac{1}{1-\alpha}\left[\mathrm{Tr}(\rho^\alpha\sigma^{1-\alpha})-1\right],
\end{align}
where we assume that $\mathrm{supp}(\rho)\subset\mathrm{supp}(\sigma)$. The monotonicity of $\mathcal{C}(\rho)$ under unital channels
follows from the facts that $D_\alpha(\rho\|\sigma)$ is monotone
under completely positive and trace preserving maps in the range $\alpha\in(0,2]$ \cite{Rastegin2016} and $D_2(\rho\|\rho_\star)=(d-1)\mathcal{C}^2(\rho)$.

\subsection*{Integrals over independent and collective local unitary transformations}
For independent local unitary transformations, the key task is to calculate
the integral $\int_{\vec{\omega}}\mathrm{Tr}(\sigma_{x_j}^jR_j\sigma_{y_j}^jR_j^\dagger)d\vec{\omega}$.
For simplicity, we define $\mathbb{P}_{i,j}=\mathrm{Tr}(\sigma_{i}R\sigma_{j}R^\dagger)$.
A straightforward observation shows that the evaluations of $\mathbb{P}_{i,j}$ can be divided into
four categories
\begin{align}
\int_{\vec{\omega}}\mathbb{P}_{i,j}d\vec{\omega}=
\left\{\begin{array}{cccc}
2, &  \mathrm{if} \quad i=0,j=0\\
0, &  \mathrm{if} \quad i=0,j\neq0\\
0, &  \mathrm{if} \quad i\neq0,j=0\\
0, &  \mathrm{if} \quad i\neq0,j\neq0\\
\end{array}\right.
\end{align}
The first three situations are evident but the last one is not immediately obvious. In order to calculate
$\int_{\vec{\omega}}\mathbb{P}_{i,j}d\vec{\omega}$, the key point is that for any $\sigma_j$
($j=1,2,3$), one can always find an operator $U\in\mathbb{SU}(2)$ such that $U\sigma_jU^\dagger=\sigma_3$.
Moreover, this operator can be absorbed into the average over $\mathbb{SU}(2)$. Therefore, we only
need to consider $\mathbb{P}_{3,3}$. Finally, we can obtain
\begin{align}
\mathbb{P}_{3,3}=\mathrm{Tr}(\sigma_{3}R\sigma_{3}R^\dagger)=2\cos^2\theta+2\sin^2\theta\cos\omega,
\end{align}
where we have used the identity
\begin{align}
R(\vec{\omega})=\exp\left(-i\frac{\omega}{2}\hat{n}\cdot\vec{\sigma}\right)=
\mathrm{e}^{-i\frac{\phi}{2}\sigma_3}\mathrm{e}^{-i\frac{\theta}{2}\sigma_2}\mathrm{e}^{-i\frac{\omega}{2}\sigma_3}
\mathrm{e}^{i\frac{\theta}{2}\sigma_2}\mathrm{e}^{i\frac{\phi}{2}\sigma_3},
\end{align}
and the commutation relations
\begin{align}
\mathrm{e}^{-i\frac{\theta}{2}\sigma_2}\sigma_3\mathrm{e}^{i\frac{\theta}{2}\sigma_2}=\cos\theta\sigma_3+\sin\theta\sigma_1,\,
\mathrm{e}^{-i\frac{\theta}{2}\sigma_2}\sigma_1\mathrm{e}^{i\frac{\theta}{2}\sigma_2}=\cos\theta\sigma_1-\sin\theta\sigma_3.
\end{align}
Finally, for any $i,j=1,2,3$ we have
\begin{align}
\int_{\vec{\omega}}\mathbb{P}_{i,j}d\vec{\omega}=\frac{1}{4\pi^2}\int_{-\pi}^\pi d\phi \int_0^\pi\sin\theta d\theta
\int_0^{2\pi}\sin^2\frac{\omega}{2} (2\cos^2\theta+2\sin^2\theta\cos\omega) d\omega=0.
\end{align}

However, the situation is much more complex when confronted with the collective local unitary transformations.
Take two-qubit state for example, and in this case we should deal with the integral
$\int_{\vec{\omega}}\mathbb{P}_{x_1,y_1}\mathbb{P}_{x_2,y_2}d\vec{\omega}$. Similar analysis leads to
the fact
\begin{align}
\int_{\vec{\omega}}\mathbb{P}_{x_1,y_1}\mathbb{P}_{x_2,y_2}d\vec{\omega}=
\left\{\begin{array}{cccc}
4, &   \quad x_1=x_2=y_1=y_2=0\\
0, &   \exists \  x_j=0 \  \mbox{or} \  y_j=0  \\
\frac{4}{3}, &  \quad x_1=x_2\neq0, y_1=y_2\neq0\\
0, &  \mathrm{otherwise}
\end{array}\right.
\end{align}
Special attention should be paid to the case $x_1=x_2\neq0, y_1=y_2\neq0$, that is
\begin{align}
\int_{\vec{\omega}}\mathbb{P}_{x_1,y_1}\mathbb{P}_{x_2,y_2}d\vec{\omega}=\frac{1}{4\pi^2}\int_{-\pi}^\pi d\phi \int_0^\pi\sin\theta d\theta
\int_0^{2\pi}\sin^2\frac{\omega}{2} (2\cos^2\theta+2\sin^2\theta\cos\omega)^2 d\omega=\frac{4}{3}.
\end{align}
We have also considered the three-qubit state and only present the result here
\begin{align}
\int_{\vec{\omega}}\mathbb{P}_{x_1,y_1}\mathbb{P}_{x_2,y_2}\mathbb{P}_{x_3,y_3}d\vec{\omega}=
&\frac{2^3}{3}(\delta_{x_1,0}\delta_{x_2\neq0,x_3}\delta_{y_1,0}\delta_{y_2\neq0,y_3}
+\delta_{x_2,0}\delta_{x_1\neq0,x_3}\delta_{y_2,0}\delta_{y_1\neq0,y_3}+\delta_{x_3,0}\delta_{x_1\neq0,x_2}\delta_{y_3,0}\delta_{y_1\neq0,y_2}) \nonumber\\
&+2^3\delta_{x_1,0}\delta_{x_2,0}\delta_{x_3,0}\delta_{y_1,0}\delta_{y_2,0}\delta_{y_3,0}+\frac{2^2}{3}\varepsilon_{x_1,x_2,x_3}\varepsilon_{y_1,y_2,y_3},
\end{align}
where $\delta_{i,j}$ and $\varepsilon_{i,j,k}$ denote the Kronecker delta and Levi-Civita symbol, respectively.
However, the generalization to arbitrary $N$-qubit states is still missing, which is left as an open question for further study.

\section*{Acknowledgements}

This research is supported by the National Natural
Science Foundation of China under Grants No. 11421063 and
No. 11534002, the National Basic Research Program of China
under Grants No. 2012CB922104 and No. 2014CB921403.

\section*{Author contributions}

Y.Y. and G.H.D. initiated the research project and established the main results under the guidance of C. P. S.
X.X. joined discussions and provided constructive suggestions.
Y.Y. wrote the manuscript and all authors commented on the manuscript.

\section*{Additional information}

\textbf{Competing financial interests:} The authors declare no competing financial interests.

\end{document}